\documentclass[letterpaper, 10 pt, conference]{ieeeconf}  
\IEEEoverridecommandlockouts                              
\usepackage{graphicx} % for pdf, bitmapped graphics files
\usepackage{times} % assumes new font selection scheme installed
\usepackage{amsmath} % assumes amsmath package installed
\usepackage{amssymb}  % assumes amsmath package installed
\usepackage{color}
\usepackage[noadjust]{cite}
\newcommand{\R}{\mathbb{R}}

\newcommand{\Lc}{\mathcal{L}}
\newcommand{\inprod}[2]{\left\langle#1, #2\right\rangle}

\DeclareMathOperator*{\argmin}{argmin}

\usepackage{amsthm}

\theoremstyle{plain}
\newtheorem{thm}{Theorem}

\newtheorem{lem}[thm]{Lemma}

\theoremstyle{definition}

\theoremstyle{remark}
\newtheorem{rem}{Remark}

\title{\LARGE \bf
Continuous-time Dynamic Realization for Nonlinear Stabilization via Control Contraction Metrics
}

\author{Ruigang Wang and Ian R. Manchester% <-this % stops a space
\thanks{This work was supported by the Australian Research Council.}
\thanks{The authors are with the Australian Centre for Field Robotics, The University of Sydney, Sydney, NSW 2006, Australia
        (e-mail: {\tt\small  ian.manchester@sydney.edu.au}).}%
}

\begin{document}

\maketitle
\thispagestyle{empty}
\pagestyle{empty}

\begin{abstract}
	 Nonlinear stabilization using control contraction metric (CCM) method usually involves an online optimization problem to compute a minimal geodesic (a shortest path) between pair of states, which is not desirable for real-time applications. This paper introduces a continuous-time dynamic realization which distributes the computational cost of the optimization problem over the time domain. The basic idea is to force the internal state of the dynamic controller to converge to a geodesic using covariant derivative information. A numerical example illustrates the proposed approach.
\end{abstract}
%%%%%%%%%%%%%%%%%%%%%%%%%%%%%%%%%%%%%%%%%%%%%%%%%%%%%%%%%%%%%%%%%%%%%%%%%%%%
\section{Introduction}

Stabilization of arbitrary trajectories of nonlinear dynamical systems is a challenging problem. One solution is to linearize the dynamics around the equilibrium manifold and apply the linear parameter-varying (LPV) control design methods \cite{Rugh:2000}. However, these approaches generally lack global stability guarantees for the closed-loop nonlinear system. Another approach is to apply nonlinear model predictive control (NMPC) \cite{Allgower:2012}, which solves an optimal control problem (OCP) in a moving horizon way. Due to the complex dynamic constraints, the computational cost often limits its applications in real-time systems.

Contraction theory \cite{Lohmiller:1998} is an attractive tool for the nonlinear stabilization problem because it provides formal global stability guarantees of the nonlinear system via simple local linear analysis. The underlying idea is to integrate the local stability results along a geodesic (a shortest path w.r.t. certain Riemannian metric) connecting the measured and reference states. Extensions to control design were developed in \cite{Manchester:2017, Manchester:2018} by introducing the concept of control contraction metric (CCM). Specifically, a CCM is a Riemannian metric for which the Riemannian energy functional of the geodesic between the measured and reference states can be made to decrease exponentially by choosing proper control action. Thus, the CCM can be understood as a differential version of control Lyapunov function (CLF). Further extensions to distributed control can be found in \cite{Wang:2017, Shiromoto:2018}. Connections and comparisons with LPV based control was discussed in \cite{Wang:2019}.

A static state-feedback realization based on integration along a geodesic was proposed in \cite{Manchester:2017}. Implementation of this controller involves solving an optimization problem to find a geodesic. This online computation is similar to NMPC, but is of lower dimension without dynamic constraints. There exist some indirect methods for geodesic computation, such as phase flow method \cite{Ying:2006}, fast marching \cite{Kimmel:1998} and graph cuts \cite{Boykov:2003}. One drawback of these approaches is the small convergence radii. Direct methods construct a finite-dimensional approximation of the online OCP and solve it via nonlinear programming (NLP). Typical discretization methods include single/multiple shooting \cite{Houska:2011} and global pseudospectral \cite{Garg:2010}. Recently, an efficient approach using the Chebyshev pseudospectral was proposed in \cite{Leung:2017}. Although the computational time is significantly reduced compared with the shooting method, online optimization is still not desirable for time-critical applications.

In this paper, we propose a continuous-time dynamic realization approach to address this issue. Inspired by a recent continuous-time MPC scheme \cite{Feller:2014,Nicotra:2018}, the proposed approach makes continuous improvements to the integral path rather than solving a full optimization problem online. Specifically, the dynamic controller use forward flows generated by the plant model and gradient information of Riemannian energy functional to force its internal state (a path connecting the reference point to the measured state) to converge to a geodesic. The control output uses the same integration technique of \cite{Manchester:2017} with integrals computed over the dynamic controller's internal state. We will consider state-feedback realization for both nominal and perturbed systems. It is shown that the nominal closed-loop system is globally exponential stable and the path converges to a geodesic if the controller dynamics are sufficiently fast with respect to the plant dynamics. For the robust case where the system is perturbed by bounded additive disturbances, one endpoint of the path would deviate from the measured state, which may lead to closed-loop instability. Robust stability is achieved by adding state feedback to the path dynamics.

The structure of the paper is as follows. Section~\ref{sec:preliminary} gives some preliminaries results on CCM-based control design. In Section~\ref{sec:realization} we detail the proposed continuous-time dynamic realization. A numerical example is presented in Section~\ref{sec:example} to illustrative the effectiveness of this approach.

%%%%%%%%%%%%%%%%%%%%%%%%%%%%%%%%%%%%%%%%%%%%%%%%%%%%%%%%%%%%%%%%%%%%%%%%%%%%
\section{Preliminaries}\label{sec:preliminary}

\subsection{Notation}
We use $ |x| $ to denote the standard Euclidean norm of a real vector $ x $. The nonnegative reals are denoted $ \R^+:=[0,\infty) $. The space $ \Lc_2^e $ is the set of vector signals $ f $ on $ \R^+ $ whose causal truncation to any finite interval $ [0,T] $ has finite squared norm, i.e. $ \sqrt{\int_{0}^{T}|f(t)|^2dt}<\infty $. For symmetric matrices $ A $ and $ B $, the notation $ A\prec B(A\preceq B) $ means that $ B-A $ is positive (semi)definite.

A Riemannian metric on $ \R^n $ is a symmetric positive-definite matrix function $ M(x) $, smooth in $ x $, which defines a smooth inner product $ \inprod{\delta_1}{\delta_2}:=\delta_1^\top M(x)\delta_2 $ for any two tangent vectors $ \delta_1,\,\delta_2 $ at the point $ x $, and the norm $ \|\delta\|_M=\sqrt{\inprod{\delta}{\delta}} $. A metric is called \emph{uniformly bounded} if $ \alpha_1I\preceq M(x) \preceq \alpha_2I,\,\forall x$, for some constants $ \alpha_2\geq \alpha_1>0 $. 

Let $ \Gamma(x,y) $ be the set of smooth paths joining two points $ x $ and $ y $ in $ \R^n $, where each $ c\in \Gamma(x,y) $ is a smooth map $ c:[0,1]\rightarrow\R^n $ and satisfying $ c(0)=x $ and $ c(1)=y $. We use the notation $ c(s),\, s\in [0,1] $ and $ c_s:=\frac{\partial c}{\partial s} $. Given a metric $ M(x) $, we can define the Riemannian length and energy functional of $ c $ as follows
\[
L(c):=\int_0^1\|c_s\|_Mds,\quad E(c):=\int_0^1\|c_s\|_M^2ds
\]
respectively. The Riemannian distance $ d(x,y) $ between two points is the length of the shortest path between them, i.e., $ d(x,y):=\inf_{c\in\Gamma(x,y)}L(c) $. Under the conditions of the Hopf-Rinow theorem, there exists a \emph{geodesic} (minimum-length curve) $ \gamma\in\Gamma(x,y) $ such that $ d(x,y)=L(\gamma) $. Furthermore, we have $ E(\gamma)=L(\gamma)^2=\inprod{\gamma_s}{\gamma_s},\,\forall s\in [0,1] $. 
%\[
%E(\gamma)=L(\gamma)^2\leq L(c)^2\leq E(c),\;\forall c\in \Gamma(x,y).
%\] 

Let $ \Gamma(x,y,t) $ be the set of smooth time-varying paths $ c:\R\times [0,1]\rightarrow\R^n $ connecting smooth signals $ x(t) $ and $ y(t) $. We also use $ c(t):=c(t,\cdot) $ and $ \dot{c}:=\frac{d c}{d t} $. The formula for first variation of energy \cite[p. 195]{Do-Carmo:1992} gives the time derivative of the energy functional $ E(t):=E(c(t)) $ as follows
\begin{equation}\label{eq:first-variation}
\frac{1}{2}\frac{dE}{dt}=\inprod{\dot{c}}{c_s}\bigr|_{s=0}^{s=1}-\int_{0}^{1}\inprod{\dot{c}}{\nabla_{c_s}c_s}ds
\end{equation}
where $ \nabla $ is the Riemannian connection induced by the metric $ M(x) $, and $ \nabla_{c_s}c_s $ is the covariant derivative. A smooth curve $ c $ is a geodesic if and only if $ \nabla_{c_s}c_s= 0 $. 

\subsection{Control Contraction Metrics}

Consider nonlinear control-affine systems of the form
\begin{equation}\label{eq:system}
\dot{x}=F(x,u):=\mathfrak{f}(x)+B(x)u
\end{equation}
where $ x(t)\in\R^n $ and $ u(t)\in \R^m $ are state and control at time $ t\in\R^+:=[0,\infty)$, respectively. For simplicity, $ \mathfrak{f} $ and $ B $ are assumed to be smooth and time-invariant. We denote the $i$th column of $ B(x) $ by $ b_i(x) $. For the system \eqref{eq:system} we define a \emph{reference trajectory} to be any set of signals $ x^*,u^* $ all in $ \Lc_2^e $ and satisfying \eqref{eq:system} on $ \R^+ $. A reference trajectory $ (x^*,u^*) $ is said to be globally exponentially stabilized by a feedback controller $ u=\kappa(x,x^*,u^*) $ if for any initial state $ x(0)\in\R^n $, a unique closed-loop solution $ x(t) $ exists for all $ t\in\R^+ $ and satisfies
\begin{equation}\label{eq:universal-statbility}
	|x(t)-x^*(t)|\leq Re^{-\lambda t}|x(0)-x^*(0)|
\end{equation}
where $ R>0 $ is the overshoot, and $ \lambda>0 $ the rate. System \eqref{eq:system} is said to be \emph{universally exponentially stabilizable} if every forward-complete solution $ (x^*,u^*) $ is globally exponentially stabilizable. Note that universal stabilizablity is a strong condition than global stabilizablity of a particular solution.

Nonlinear stabilization using control contraction metric (CCM) (\cite{Manchester:2017}) is a constructive approach to achieve universal stability. For the offline design stage, it applies linear system theory to the control synthesis of the local linearized system -- \emph{differential dynamics}:
\begin{equation}\label{eq:diff-dyn}
\dot{\delta}_x=A(x,u)\delta_x+B(x)\delta_u
\end{equation} 
where $ A=\frac{\partial \mathfrak{f}}{\partial x}+\sum_{i=1}^m \frac{\partial b_i}{\partial x}u_i$. Specifically, we construct a differential feedback law:
\begin{equation}\label{eq:diff-control}
\delta_u=K(x)\delta_x
\end{equation} 
where $ K=YW^{-1} $ with $ W(x)\in\R^{n\times n} $ and $ Y(x)\in\R^{m\times n} $ obtained from the following parameter-dependent linear matrix inequality (LMI):
\begin{equation}\label{eq:ccm-synthesis}
-\dot{W}+AW+WA^\top-BY-Y^\top B^\top+2\lambda W\preceq 0.
\end{equation} 
From the above inequality, the controller \eqref{eq:diff-control} achieves exponential stability for \eqref{eq:diff-dyn}:
\begin{equation}\label{eq:diff-stability}
\frac{d}{dt}\|\delta_x\|_M^2=\delta_x^\top\dot{M}\delta_x+2\delta_x^\top M(A+BK)\delta_x \leq -2\lambda \|\delta_x\|_M^2
\end{equation}
where $ M(x)=W^{-1}(x) $ is called a CCM. 

A static (memoryless) realization of the controller was proposed in \cite{Manchester:2017}, which includes three steps:
\begin{enumerate}
	\item Compute a minimal geodesic 
	\begin{equation}\label{eq:geodesic-computation}
	\gamma(t):=\argmin_{c\in\Gamma(x^*(t),x(t))}E(c).
	\end{equation}
	\item Integrate \eqref{eq:diff-control} over $ \gamma $, i.e.,
	\begin{equation}\label{eq:realization-ideal}
	\kappa_\gamma(t,s):=u^*(t)+\int_0^sK(\gamma(t,\mathfrak{s}))\gamma_s(t,\mathfrak{s}) d\mathfrak{s}.
	\end{equation}
	\item Implement the state-feedback control 
	\begin{equation}\label{eq:static-geo}
		u=\kappa(x,x^*,u^*):=\kappa_\gamma(t,1).
	\end{equation}
\end{enumerate}
This static realization achieves universally exponential stability with overshoot $ R=\sqrt{\frac{\alpha_2}{\alpha_1}} $ and rate $ \lambda $. If there exists a smooth coordinate transformation $ z=h(x) $ satisfying $ \delta_z^\top\delta_z=\delta_x^\top M(x)\delta_x $, we can compute the geodesics directly via $ \gamma(s)=h^{-1}(z^*(1-s)+zs) $ where $ z^*=h(x^*) $ and $ z=h(x) $. However, for general cases, the computation of geodesics involves an optimization problem \eqref{eq:geodesic-computation}, which is not desired for time-critical applications. 
 
%%%%%%%%%%%%%%%%%%%%%%%%%%%%%%%%%%%%%%%%%%%%%%%%%%%%%%%%%%%%%%%%%%%%%%%%%%%%

%%%%%%%%%%%%%%%%%%%%%%%%%%%%%%%%%%%%%%%%%%%%%%%%%%%%%%%%%%%%%%%%%%%%%%%%%%%%
\section{Continuous-time Dynamic Realization}\label{sec:realization}

In this section, we introduce a continuous dynamic control realization which keeps Step 2) and 3) unchanged but replace Step 1) with a dynamical system whose internal state is a path joining $ x^*(t) $ to $ x(t) $. This path dynamics solves a geodesic computation problem in parallel with the plant system. As shown in Fig.~\ref{fig:dynamic-control}, we will consider two scenarios: nominal and robust state feedback. In particular, for the robust case, we assume that the system \eqref{eq:system} is perturbed by bounded additive disturbances, i.e.,
\begin{equation}\label{eq:system-perturbed}
\dot{x}=F(x,u)+d
\end{equation}
with $ \|d(t)\|\leq \Delta $ for all $ t\in\R^+ $.

\begin{figure}
	\centering
	\begin{tabular}{cc}
		\includegraphics[width=0.46\linewidth]{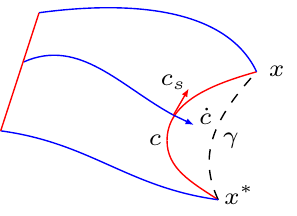} &
		\includegraphics[width=0.46\linewidth]{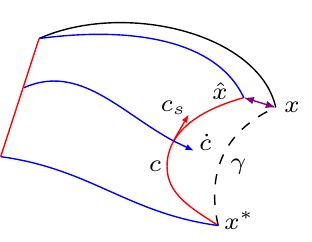} \\
		{\scriptsize (a) Nominal case } & {\scriptsize (b) Robust case } 
	\end{tabular}
	\caption{Geometric illustrations of the continuous-time dynamic realization: red -- path $ c(t,\cdot) $, blue -- flows $ c(\cdot,s) $, dash -- geodesic $ \gamma(t) $.}
	\label{fig:dynamic-control}
\end{figure}

\subsection{Nominal State Feedback}
First, we consider the continuous-time dynamic realization via the forward flow defined by \eqref{eq:system}:
\begin{equation}\label{eq:pd-forward}
\begin{split}
\dot{c}&=f(t,s):=F(c(t,s),\kappa_c(t,s)) \\
u&=\kappa_c(t,1) 
\end{split}
\end{equation}
where the initial state is $ c(0,s)=sx(0)+(1-s)x^*(0) $. Then the endpoint dynamics can be represented by 
\begin{subequations}
	\begin{align}
	\dot{c}(t,0)&=F(x^*,u^*), \label{eq:reference}\\
	\dot{c}(t,1)&=f(x,u).
	\end{align}
\end{subequations}
It is easy to verify that $ c(t,0)=x^*(t) $ and $ c(t,1)=x(t) $ for all $ t\in\R^{+} $ since $ c(0,0)=x^*(0) $ and $ c(0,1)=x(0) $. Moreover, integration of \eqref{eq:diff-stability} over $ c(t) $ gives
\begin{equation}\label{eq:path-stability}
\frac{1}{2}\dot{E}=\inprod{f(t,s)}{c_s}\bigr|_{s=0}^{s=1}-\int_{0}^{1}\inprod{f(t,s)}{\nabla_{c_s}c_s}ds\leq -\lambda E.
\end{equation}
Globally exponential stability is achieved but perhaps with larger overshoot since $ c(t) $ generally does not converge to a geodesic $ \gamma(t) $. 

Now we consider an alternative path dynamics:
\begin{equation}\label{eq:pd-nominal}
\begin{split}
\dot{c}=f(t,s)+\alpha(s)\nabla_{c_s}c_s
\end{split}
\end{equation}
where $  \alpha:[0,1]\rightarrow\R^+ $ is a smooth weighting function satisfying $ \alpha(0)=\alpha(1)=0 $. Here the covariant derivative $ \nabla_{c_s}c_s $ can be taken as the gradient information of the geodesic optimization problem \eqref{eq:geodesic-computation}. We define the normalized weighting function as $ \eta(s)=\alpha(s)/\overline{\alpha} $ with $ \overline{\alpha}=\max_{s\in[0,1]}\alpha(s) $.  Then, the nominal stability is given as follows.
\begin{thm}\label{thm:nominal}
	For any weighting function $ \alpha(s) $, the system \eqref{eq:system} subject to the control law \eqref{eq:pd-nominal} is universally exponentially stable. If the parameter $ \overline{\alpha} $ is chosen to be sufficiently large, then the controller internal state $ c(t,\cdot) $ converges to a geodesic $ \gamma(t,\cdot)\in\Gamma(x^*,x,t) $ before $ x(t) $ converges to $ x^*(t) $.
\end{thm}
\begin{proof}	
From \eqref{eq:path-stability}, we have
\begin{equation}\label{eq:universal-stability}
\frac{1}{2}\dot{E}\leq -\lambda E-\overline{\alpha}\int_{0}^{1}\eta(s)\|\nabla_{c_s}c_s\|_M^2ds\leq -\lambda E.
\end{equation}
Universal stability of the nominal system \eqref{eq:system} follows as the length of $ c(t) $ shrinks exponentially. 

Choose a constant $ \tau\in (0,1) $, from Lemma~\ref{lem:1} we can online adjust the weighting function $ \eta $ such that the following inequality holds:
\begin{equation}
\frac{1}{2}\dot{E}\leq -\lambda E-\overline{\alpha}\tau\int_{0}^{1}\|\nabla_{c_s}c_s\|_M^2ds.
\end{equation} 
With a sufficiently large $ \overline{\alpha} $, the closed-loop system can be decomposed into two time-scale subsystems: slow dynamics \eqref{eq:system} and fast dynamics \eqref{eq:pd-nominal}. The covariant derivative $ \nabla_{c_s}c_s $ will be forced to converge to 0 (i.e., the path $ c(t,\cdot) $ converges to a geodesic $ \gamma(t,\cdot)\in\Gamma(x^*,x,t) $) before the convergence of the state $ x(t) $ to $ x^*(t) $. 
\end{proof}
\begin{rem}
	As shown in Section~\ref{sec:implementation}, the online implementation only computes a finite number of flows $ c(t,s_j),\,j=0,1,\ldots,N $ digitally using forward-Euler or Runge-Kutta methods with a sufficiently small sampling time $ \tau_s $. Thus, $ \overline{\alpha} $ cannot be chosen to be arbitrary large due to numerical stability consideration and the parameters $ (s_0,s_1) $ for the weighting function $ \eta(s) $ in \eqref{eq:weighting-fun} cannot be chosen to be arbitrary close to $ (0,1) $. Although the path $ c(t,\cdot) $ may not follow $ \gamma(t,\cdot) $ exactly, it can still converge to a small neighborhood of $ \gamma(t,\cdot) $.
\end{rem}
%%%%%%%%%%%%%%%%%%%%%%%%%%%%%%%%%%%%%%%%%%%%%%%%%%%%%%%%%%%%%%%%%%%%%%%%%%%%
\subsection{Robust State Feedback}
When system \eqref{eq:system} is perturbed by external disturbances, the state trajectory $ x(\cdot) $ generally does not coincide with the endpoint trajectory $ c(\cdot,1) $ generated by the path dynamics \eqref{eq:pd-nominal}. Let $ \hat{x}(t)=c(t,1)$ and $ \tilde{x}(t)=x(t)-\hat{x}(t) $. To reduce the disturbance effect on $ \tilde{x} $, we use the following path dynamics
\begin{equation}\label{eq:realization-robust}
\begin{split}
	\dot{c}=f(t,s)+\alpha(s)\nabla_{c_s}c_s+\beta(s)\tilde{x}(t)
\end{split}
\end{equation}
where $ \beta(s)=\overline{\beta}\zeta(s) $ with $ \zeta:[0,1]\rightarrow [0,1] $ as a nondecreasing function satisfying $ \zeta(0)=0 $ and $ \zeta(1)=1 $. Note that, for the nominal case, the above system is equivalent to the path dynamics \eqref{eq:pd-nominal}.

If the disturbance bound $ \Delta $ is sufficiently small, the dynamics of $ \tilde{x}(\cdot) $ can be approximated by
\begin{equation}\label{eq:diff-dyn-approximation}
	\dot{\tilde{x}}=(A_{cl}(\hat{x},\hat{u})-\overline{\beta}I)\tilde{x}+d
\end{equation}
where $ A_{cl}(\hat{x},\hat{u})=A(\hat{x},\hat{u})+B(\hat{x})K(\hat{x}) $ with $ \hat{u}=\kappa_c(t,1) $. From \eqref{eq:diff-stability} we have that the maximum eigenvalue of $ A_{cl}(\hat{x},\hat{u}) $ is no larger than $ -\lambda $. Therefore, the error bound for $ \tilde{x} $ is 
\begin{equation}\label{eq:dx-bound}
	|\tilde{x}(t)|\leq \frac{\Delta}{\overline{\beta}+\lambda},\quad\forall t\in\R^+.
\end{equation}
The time derivative of the energy functional satisfies
\begin{equation}\label{eq:robust-proof-1}
	\begin{split}
	\frac{1}{2}\dot{E}\leq & -\lambda E+\inprod{\overline{\beta}\tilde{x}}{c_s(t,1)}-\int_{0}^{1}\inprod{\overline{\beta}\tilde{x}}{\zeta\nabla_{c_s}c_s}ds \\
	&-\overline{\alpha}\tau\int_{0}^{1}\|\nabla_{c_s}c_s\|_Mds  \\
	\leq &  -\lambda E +\inprod{\overline{\beta}\tilde{x}}{c_s(t,1)}+\frac{\epsilon^2}{4}\|\overline{\beta}\tilde{x}\|_M^2 \\
	& -\frac{\overline{\alpha}\tau}{2}\int_{0}^{1}\|\nabla_{c_s}c_s\|_Mds
	\end{split}
\end{equation} 
where $ \epsilon\geq \sqrt{2/\overline{\alpha}\tau} $. The closed-loop robust stability is given as follows.
\begin{thm}
	Consider the perturbed system \eqref{eq:system-perturbed} and the continuous-time dynamic control realization \eqref{eq:realization-robust}. If the parameter $ \overline{\alpha},\overline{\beta} $ are sufficiently large, the closed-loop system is robust stable with respect to the set 
	\begin{equation}\label{eq:pd-robust}
		\Omega(x^*)=\{x\in\R^n:|x-x^*|\leq \overline{R}\Delta/\lambda\}
	\end{equation}
	where $ \overline{R}=\frac{(1+\sqrt{1+\lambda\epsilon^2})\overline{\beta}}{2(\overline{\beta}+\lambda)}R+\frac{\lambda}{\overline{\beta}+\lambda} $.
\end{thm}
\begin{proof}
	If $ \overline{\alpha} $ is sufficiently large, we can conclude from \eqref{eq:robust-proof-1} that $ c(t) $ converges to $ \gamma(t) $ for $ t\geq T $ where $ T $ is sufficiently large. This leads to
	\begin{equation}
		\frac{1}{2}\dot{E}\leq -\lambda E+\inprod{\overline{\beta} \tilde{x}}{\gamma_s(t,1)}+\frac{\epsilon}{4}\|\overline{\beta}\tilde{x}\|_M^2.
	\end{equation}
	With the facts that $ E(\gamma)=\inprod{\gamma_s}{\gamma_s},\forall s\in[0,1] $ and $ |\tilde{x}|\leq \frac{\Delta}{\overline{\beta}+\lambda} $, there exists a $ T'>T $ such that the following inequality holds for $ t\geq T' $:
	\begin{equation}
		\begin{split}
		|\hat{x}(t)-x^*(t)|\leq \frac{(1+\sqrt{1+\lambda\epsilon^2})\overline{\beta}}{2(\overline{\beta}+\lambda)}R\Delta
		\end{split}
	\end{equation}
	which leads to
	\begin{equation}
		|x(t)-x^*(t)|\leq |\hat{x}(t)-x^*(t)|+|\tilde{x}|\leq \overline{R}\Delta/\lambda.
	\end{equation}
\end{proof}
\begin{rem}
	The parameter $ \overline{\alpha} $ represents the convergence speed of $ c(t) $ to a geodesic $ \gamma(t) $ connecting $ x^*(t) $ to $ \hat{x}(t) $ while the parameter $ \overline{\beta} $ controls the convergence speed of $ \hat{x}(t) $ to $ x(t) $. The parameter $ \epsilon $ affects the size of invariant set. When $ \overline{\alpha},\overline{\beta}\rightarrow \infty $ and $ \epsilon\rightarrow 0 $, we have $ \overline{R}\rightarrow R $ which implies that dynamic realization achieves the same invariant set as the geodesic based static realization \eqref{eq:static-geo}.
\end{rem}
 
%%%%%%%%%%%%%%%%%%%%%%%%%%%%%%%%%%%%%%%%%%%%%%%%%%%%%%%%%%%%%%%%%%%%%%%%%%%% 
\subsection{Implementation}\label{sec:implementation}

\begin{figure}[!bt]
	\centering
	\includegraphics[width=1\linewidth]{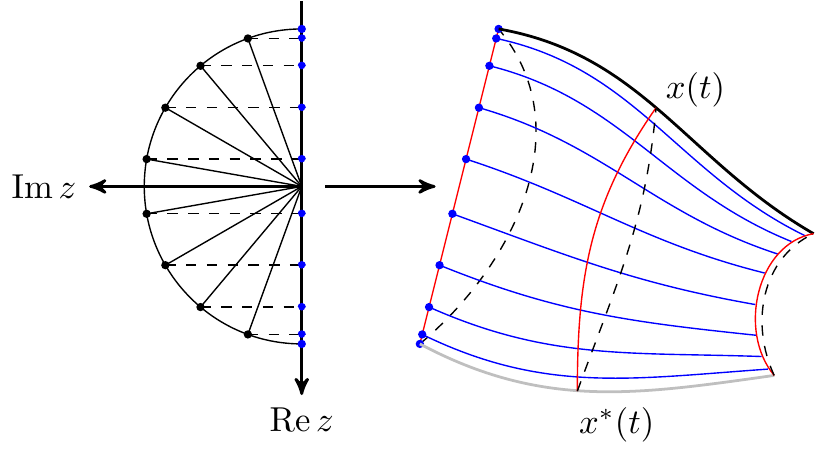}
	\caption{Illustration of discretization of path $ c(t) $ using Chebyshev polynomials.}\label{fig:implementation}
\end{figure}

The path dynamics is an infinite-dimensional system as its internal state $ c(t) $ is a smooth function over $ [0,1] $. For online implementation, we approximate $ c(t) $ with Chebyshev polynomial expansion at time $ t $. This is a finite-dimensional approximation based on the samples of the path at Chebyshev nodes. In this way, the path dynamics is discretized into a finite set of dynamical systems whose state dimension is same as the original nonlinear plant. Those systems are solved in parallel with the nonlinear plant and the solutions are used to construct an approximate path at the next time step. As time involves, this path converges to a geodesic due to the forward and gradient descent flows. This approach is different from \cite{Leung:2017} which uses Chebyshev polynomials to discretize the geodesic computation problem \eqref{eq:geodesic-computation} at each time point. A finite-dimensional NLP is iteratively solved online, and the optimal solution is then used to construct a geodesic. Thus, the online computation time of the proposed approach is expected to be much smaller, compared with the optimization based approach \cite{Leung:2017}.

Firstly, we recall some standard results of approximation theory using Chebyshev polynomials (see \cite{Trefethen:2013} for details). The first-kind Chebyshev polynomials $ T_k(x) $ over the interval $ [-1,1] $ are defined recursively by
\begin{equation}
	T_{k+1}(x)=2xT_{k}(x)-T_{k-1}(x),\quad k=1,2,3,\ldots
\end{equation}
with starting values $ T_0(x)=1$ and $ T_1(x)=x $. Under the coordinate transform $ x=\cos(\theta) ,\, \theta\in[0,\pi] $, we have $ T_k(\cos\theta)=\cos(k\theta) $. A continuous function $ f(x) $ over the interval $ [-1,1] $ can be approximated by
\begin{equation}\label{eq:chebyshev}
\tilde{f}(\cos\theta)=\frac{a_0}{2}+\sum_{k=1}^{N}a_k\cos(k\theta)
\end{equation}
where the coefficients $\{ a_k\}_{0\leq k\leq N} $ can be obtained by apply discrete cosine transform to the samples of $ f $ at the Chebyshev nodes $ x_j=\cos(j\pi/N),j=0,1,\ldots,N $. 
Other operations on $ f $ such as integration and differentiation can also be efficiently approximated by Chebyshev polynomials.  

Since the CCM-based control design is invariant under coordinate transformations \cite{Manchester:2017}, we can reparameterize the path $ c(t,\cdot) $ from $ [0,1] $ to $ [-1,1] $. For the online implementation of dynamic controllers \eqref{eq:pd-forward}, \eqref{eq:pd-nominal} and \eqref{eq:pd-robust}, instead of computing the infinite-dimensional state $ c(t,\cdot) $, we only compute the flows at $ s_j=\cos(j\pi/N),j=0,1,\ldots,N $, as shown in Fig.~\ref{fig:implementation}. Base on the values of $ c(t,s_j) $, the state $ c(t,\cdot) $ is reconstructed as $ c(t,s)=\mathbf{c}(t)T(s) $ where $ \mathbf{c}(t)\in\R^{n\times (N+1)} $ and $ T(s)=[T_0(s),T_1(s),\ldots,T_N(s)]^\top $. With this, we can compute smooth representations of the derivative $ \frac{\partial c}{\partial s} $, covariant derivative $ \nabla_{c_s}c_s $, differential control $ \delta_u $ and its integration $ \kappa_c $. By taking samples of these functions at the Chebyshev nodes, we can obtain the right hand side of the dynamic controllers \eqref{eq:pd-forward}, \eqref{eq:pd-nominal} and \eqref{eq:pd-robust}. 

\begin{rem}
	Note that the above computation only involves a series of simple online operations such as additions, multiplications, differentiation and integration over a smooth function $ c(t,\cdot) $. Due to the absence of complex operations (e.g. solutions of optimization problems) the online computational time can be estimated \emph{a priori}. This information can be used to choose a sufficiently small sampling time $ \tau_s $ such that the flows $ c(\cdot,s_j),j=0,1,\ldots,N $ can be computed digitally by using forward-Euler or Runge-Kutta approximation methods.
\end{rem}

%%%%%%%%%%%%%%%%%%%%%%%%%%%%%%%%%%%%%%%%%%%%%%%%%%%%%%%%%%%%%%%%%%%%%%%%%%%% 
\section{Illustrative Example}\label{sec:example}

\begin{figure}[!bt]
	\centering
	\includegraphics[width=0.8\linewidth]{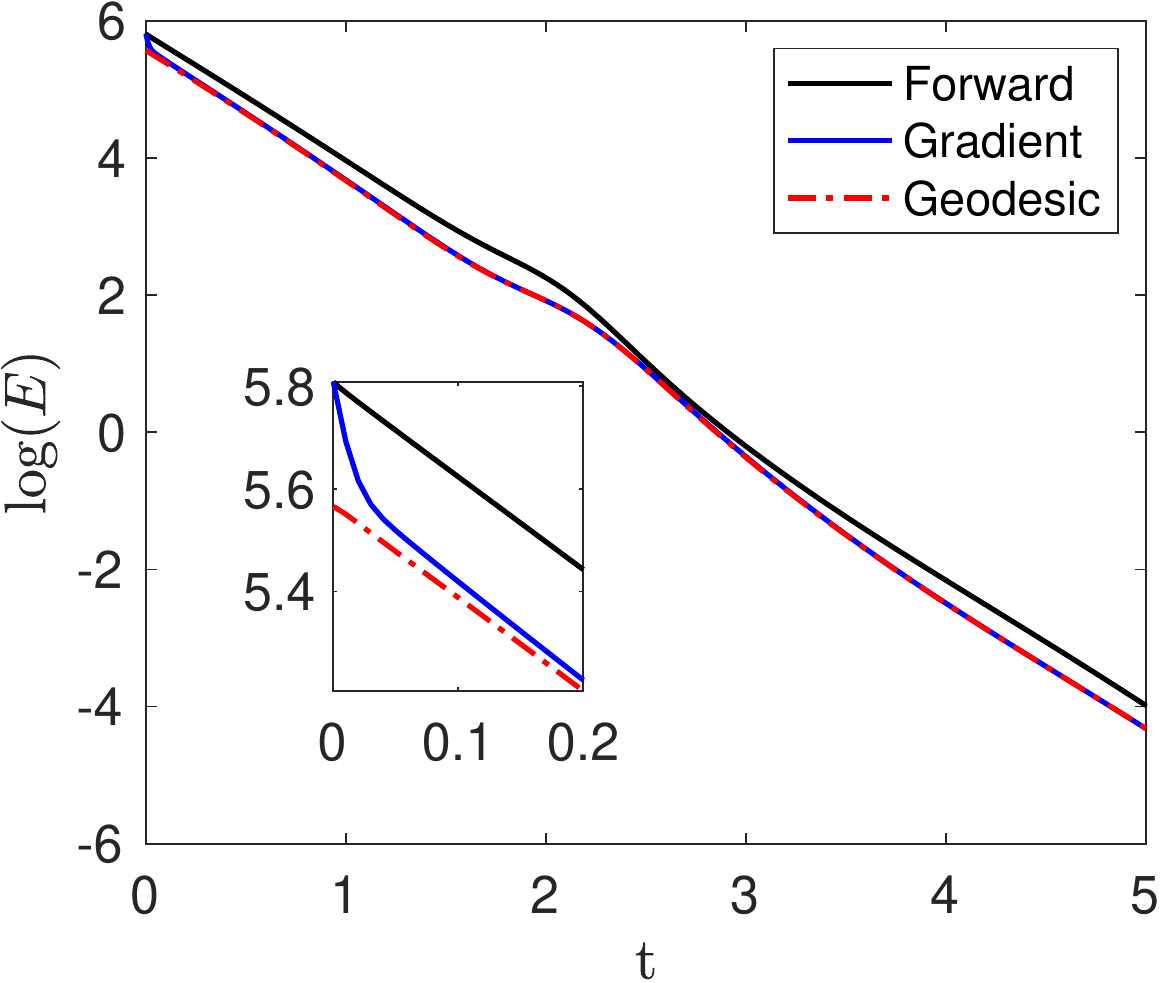}
	\caption{Exponential decay rate of the Riemannian energy of integral paths.}\label{fig:re-nm}
\end{figure}

\begin{figure*}[!tb]
	\centering
	\includegraphics[width=0.9\linewidth]{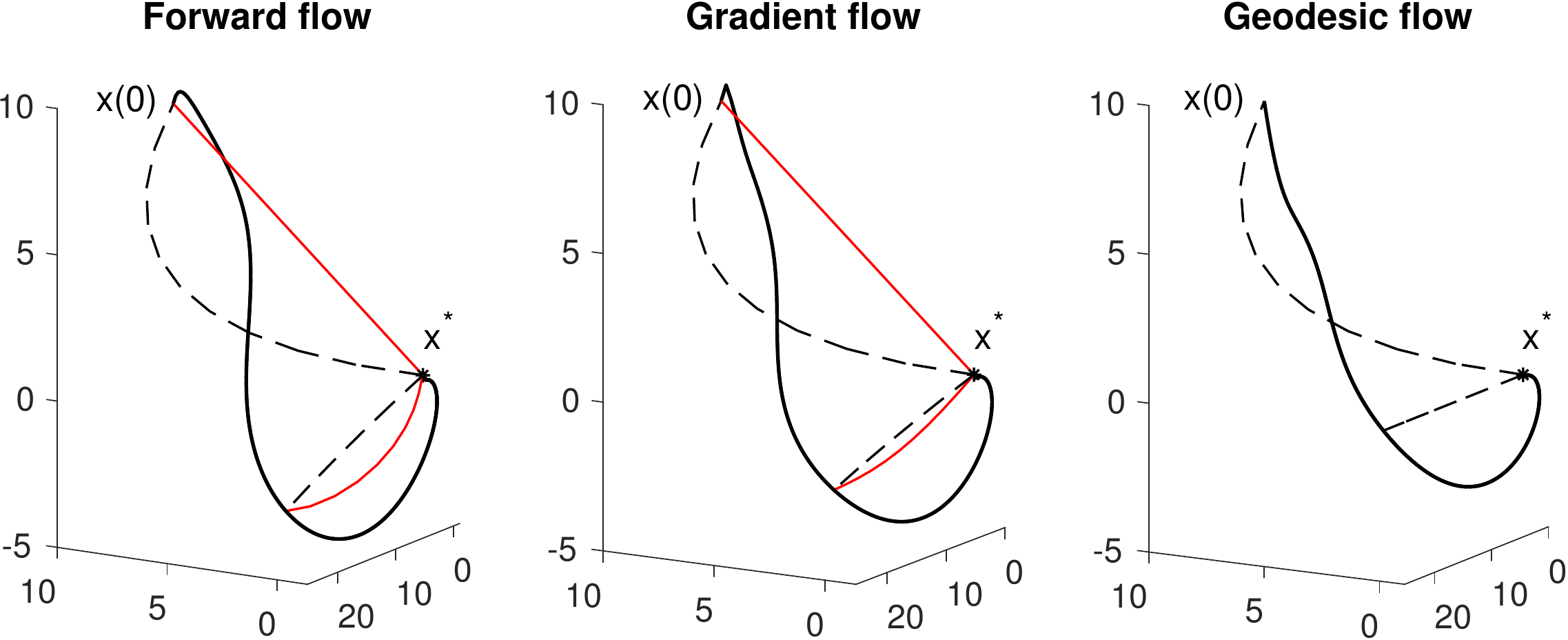}
	\caption{Nominal state-feedback control: red -- path $ c(t,\cdot) $ at two time points, black dash -- geodesic $ \gamma(t,\cdot) $, black solid -- state trajectory $ x(t) $. }\label{fig:cmp-nm}
\end{figure*}

\begin{figure*}[!tb]
	\centering
	\includegraphics[width=0.9\linewidth]{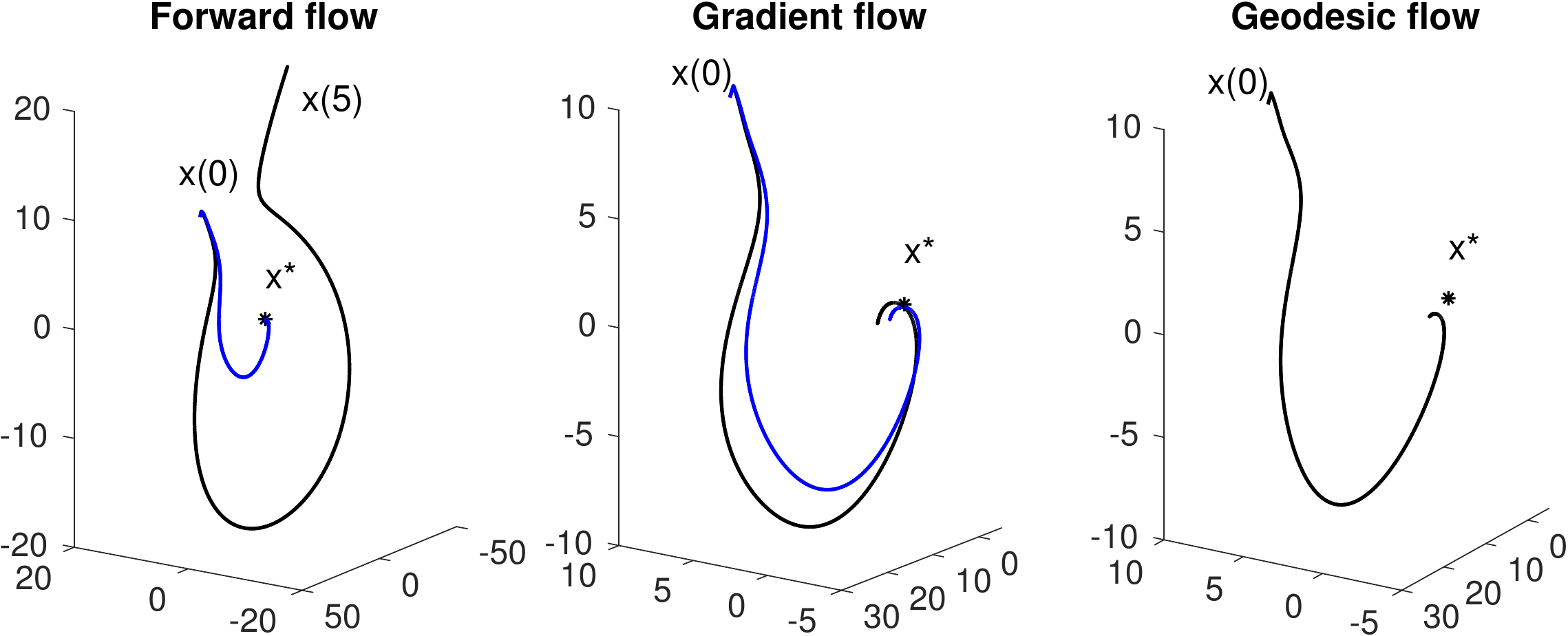}
	\caption{Robust state-feedback control: blue -- predicted state trajectory $ c(t,1) $, black solid -- state trajectory $ x(t) $.}\label{fig:cmp-rb}
\end{figure*}

We consider the following nonlinear system
\begin{equation}
	\begin{bmatrix}
		\dot{x}_1 \\ \dot{x}_2 \\ \dot{x}_3
	\end{bmatrix}=
	\begin{bmatrix}
	-x_1+x_3 \\
	x_1^2-x_2-2x_1x_2+x_3 \\
	-x_2
	\end{bmatrix}+
	\begin{bmatrix}
	0 \\ 0 \\ 1
	\end{bmatrix}u.
\end{equation}
This system is not feedback linearizable and highly unstable. The control synthesis problem \eqref{eq:ccm-synthesis} was solved by SOS programming with LMI toolbox - Yalmip \cite{Lofberg:2004}. A control contraction metric with $ \lambda=1 $ was found to be
\[
W(x)=W_0+W_1x_1+W_2x_1^2
\]
where 
\begin{gather*}
W_0=
\begin{bmatrix}
2.686  &  0.237 &  -1.816 \\
0.237 &  16.265 &   2.006 \\
-1.816  &  2.006 &   6.395
\end{bmatrix} \\
W_1=
\begin{bmatrix}
     0  & -5.373 & 0 \\
-5.373 &   -0.948 &    3.631 \\
0  &  3.631 &   0
\end{bmatrix},
W_2=
\begin{bmatrix}
0  & 0 & 0 \\
0  & 10.747 &  0 \\
0  & 0 &  0
\end{bmatrix}
\end{gather*}
and $ Y(x)=-\frac{1}{2}\rho(x)B^\top $ with $ \rho(x)=19.614+1.386x_1+9.616x_1^2 $. From \cite[Lemma 1]{Manchester:2017}, this metric is complete and thus a minimal geodesic exists for every pair of points. For online implementation, we use Chebyshev basis functions with maximal order $ N=4 $ to reconstruct the path $ c(t,\cdot) $. The toolbox for Chebyshev polynomials manipulation is called \emph{chebfun} \cite{Driscoll:2014}, which is an open source software.

For the nominal case, we compare the results of three different realizations: forward flow based dynamic controller \eqref{eq:pd-forward}, gradient flow based dynamic controller \eqref{eq:pd-nominal} and geodesic based static controller \eqref{eq:static-geo}. The initial condition and setpoint are chosen as $ x(0)=[9,9,9]^\top $ and $ x^*=[0,0,0]^\top $, respectively. From Fig.~\ref{fig:re-nm}, the Riemannian energy functional of the integral path $ c $ decays exponentially with rate of $ 2\lambda $ for these three controllers. The proposed approach converges to a geodesic within time of $ 0.05 $ by feeding the covariant derivative to the path dynamics. Without this term, the path in forward flow based approach does not converge to a geodesic, which leads to a larger overshoot estimation for exponential stability.  Fig.~\ref{fig:cmp-nm} depicts the time evolution of integral paths $ c(t,\cdot) $ and geodesics $ \gamma(t) $ for different controllers. Compared with the forward flow approach, given the same initial condition (a straight line), the state $ c(t,\cdot) $ of the proposed approach converges to the neighborhood of a geodesic $ \gamma(t) $. 

For the robust case where the dynamics of $ x_1 $ is perturbed by a persistent external disturbance $ d(t)=2 $, we test those three controllers using the same initial state and setpoint. Fig.~\ref{fig:cmp-rb} shows that the forward flow approach is unstable due to the lack of feedback, although the state prediction $ \hat{x}(t) $ converges to the setpoint. For the proposed approach, the state prediction $ \hat{x}(t) $ remains in a neighborhood of $ x(t) $ due to the feedback term in \eqref{eq:pd-robust}. And the closed-loop system has a similar response compared to the geodesic based approach.

%%%%%%%%%%%%%%%%%%%%%%%%%%%%%%%%%%%%%%%%%%%%%%%%%%%%%%%%%%%%%%%%%%%%%%%%%%%% 
\section{Conclusion}
In this paper we proposed a continuous-time dynamic realization for control contraction metrics based nonlinear stabilization. It distributes the online geodesic computation across the time domain. Both universal stability for the nominal system and robust stability for the perturbed system are guaranteed. Simulation results demonstrated the effectiveness of the proposed approach.
%\addtolength{\textheight}{-12cm}  

\section*{APPENDIX}
\begin{lem}\label{lem:1}
	For any $ c\in\Gamma(x^*,x) $ and any $ \tau\in(0,1) $, there exists a weighting function $ \eta:[0,1]\rightarrow [0,1] $ such that 
	\begin{equation}\label{eq:weighting-function}
	\begin{split}
	\int_{0}^{1}\eta(s)\|\nabla_{c_s}c_s\|_M^2ds\geq \tau\int_{0}^{1}\|\nabla_{c_s}c_s\|_M^2ds.
	\end{split}
	\end{equation}
\end{lem}
\begin{proof}
	We define $ \mu_c(s):=\int_{0}^{s}\|\nabla_{c_s}c_s\|_M^2d\mathfrak{s} $. Since $ M(x) $ is a uniformly bounded metric and $ c $ is a smooth curve, the covariant derivative $ \nabla_{c_s}c_s $ is smooth and bounded for any $ s\in[0,1] $. Thus, $ \mu_c $ is a nondecreasing function with $ \mu_c(0)=0 $ and $ \mu_c(1)=C<\infty $. If $ c $ is a geodesic (i.e., $ C=0 $), the weighting function $ \eta(s)=0 $ satisfies \eqref{eq:weighting-function}. Otherwise, for any $ \tau\in (0,1) $, we can find $ s_0=\arg\inf\mu_c^{-1}((1-\tau)C/2) $ and $ s_1=\arg\sup \mu_c^{-1}((1+\tau)C/2) $. It is easy to check that $ 0<s_0<s_1<1 $ and $ \mu_c(s_1)-\mu_c(s_0)\geq \tau C $ since $ \mu_c $ is nondecreasing. Now we choose the weighting function to be 
	\begin{equation}\label{eq:weighting-fun}
	\eta(s)=\pi_0^{s_0}(s)\left[1-\pi_{s_1}^1(s)\right]
	\end{equation} 
	where $ \pi_a^b: \R\rightarrow [0,1] $ is a smooth and nondecreasing such that $ \pi_a^b(s)=0,\forall s\leq a $ and $ \pi_a^b(s)=1,\forall s\geq b $. This weighting function satisfies \eqref{eq:weighting-function} as
	\[
	\begin{split}
	\int_{0}^{1}\eta(s)\|\nabla_{c_s}c_s\|_M^2ds \geq \int_{s_0}^{s_1}\|\nabla_{c_s}c_s\|_M^2ds \geq \tau C.
	\end{split}
	\] 
\end{proof}

%\section*{ACKNOWLEDGMENT}

% References
\bibliographystyle{IEEEtran}
\bibliography{ref}

\end{document}